\documentclass[a4paper,11pt]{article}
\usepackage{amsmath, amssymb}
\usepackage{mathenv,amsthm}
\usepackage[boxed]{algorithm2e}
\usepackage{multirow}
\usepackage[margin=2.5cm]{geometry}

\newtheorem{theorem}{Theorem}[section]

\newtheorem{lemma}[theorem]{Lemma}

\newenvironment{remark}[1][Remark :]{\begin{trivlist}
\item[\hskip \labelsep {\bfseries #1}]}{\end{trivlist}}
\newcommand{\F}{\mathbb{F}}

\title{Refinements of Miller's Algorithm over Weierstrass Curves Revisited}

\author{Duc-Phong Le\\ \small Temasek Laboratories\\[-0.8ex] \small National University of Singapore\\[-0.8ex] \small \texttt{tslld@nus.edu.sg} \and Chao-Liang Liu\\ \small Department of Applied Informatics and Multimedia\\[-0.8ex] \small Asia University.\\[-0.8ex] \small \texttt{jliu@asia.edu.tw} }

\date{}
\begin{document}

\maketitle
%\shortauthors{D.-P. Le and C.-L. Liu}

%\keywords{Miller's algorithm; Weil pairing; Tate pairing; Efficient Computation; Weierstrass Curves.}

\begin{abstract}
In 1986 Victor Miller described an algorithm for computing the Weil pairing in his unpublished manuscript. This algorithm has then become the core of all pairing-based cryptosystems. Many improvements of the algorithm have been presented. Most of them involve a choice of elliptic curves of a \emph{special} forms to exploit a possible twist during Tate pairing computation. Other improvements involve a reduction of the number of iterations in the Miller's algorithm. 

For the generic case, Blake, Murty and Xu proposed three refinements to Miller's algorithm over Weierstrass curves. Though their refinements which only reduce the total number of vertical lines in Miller's algorithm, did not give an efficient computation as other optimizations, but they can be applied for computing \emph{both} of Weil and Tate pairings on \emph{all} pairing-friendly elliptic curves. In this paper we extend the Blake-Murty-Xu's method and show how to perform an elimination of all vertical lines in Miller's algorithm during Weil/Tate pairings computation on \emph{general} elliptic curves. Experimental results show that our algorithm is faster about 25\% in comparison with the original Miller's algorithm.
\end{abstract}

\maketitle

\section{Introduction}
In recent years, the Weil/Tate pairings and their variants have become extremely useful in cryptography. The first notable application of pairings to cryptology was the work of Menezes, Okamato and Vanstone~\cite{MOV91}  who showed that the discrete logarithm problem on a supersingular Elliptic Curve can be reduced to the discrete logarithm problem in a Finite Field in 1991 due to the Weil pairing. Frey and R\"uck~\cite{FR94} also consider this situation using the Tate pairing. 
%Pairings have first been suggested as a means of attacking cryptosystems by computing discrete logarithms on certain curves.
However, the applications of pairings in constructing cryptographic protocols has only attracted attention after Joux' seminal paper describing an one-round 3-party Diffie-Hellman key exchange protocol~\cite{Jou00} in 2000. Since then, the use of cryptosystems based on pairings has had a huge success with some notable breakthroughs such as the first practical Identity-based Encryption (IBE) scheme~\cite{BF01}, the short signature scheme~\cite{BLS01} from Weil pairing.

The efficient algorithms for Weil/Tate parings computation thus play a very important role in pairing-based cryptography. The best known method for computing Weil/Tate pairings is based on Miller's algorithm~\cite{Mil04} for rational functions from scalar multiplications of divisors. The Weil pairing requires two \textit{Miller loops}, while the Tate pairing requires only one application of the \textit{Miller loop} and a \textit{final exponentiation}. 

Consequently, many improvements on Miller's algorithm presented are based in some manner on it. Barreto et al.~\cite{BLS03} pointed out that we can ignoring all terms that are contained in a proper subfield of $\mathbb{F}_{p^k}$ during the computation of Tate pairing when the elliptic curves chosen have the even \emph{embedding degree}\footnote{A subgroup $G$ of the group of points of an elliptic curve $E(\mathbb{F}_q)$ is said to have {\it embedding degree} $k$ if its order $n$ divides $q^k - 1$, but does not divide $q^i - 1$ for all $0 < i < k$.}. Another approach of improving the algorithm is to reduce the Miller-loop length by introducing variants of Tate pairings such as Eta pairing~\cite{BGHS07}, Ate pairing~\cite{HSV06,LLP09}, and  optimal pairings~\cite{Ver10, Hes08}.

For a more generic approach, Blake, Murty and Xu~\cite{BMX06} proposed three refinements to Miller's algorithm. Their refinements allowed to reduce the total number of vertical lines in Miller's algorithm thanks to an elegant observation involving \textit{conjugate} of a linear function $h(x, y) = k(x - a) + b - y$\footnote{The equation of the conjugate of $h$, denoted $\bar{h}(x, y)$ is $\bar{h}(x, y) = k(x - a) + b + y + a_1 x + a_3$, where $a_1, a_3$ are parameters of an elliptic curve of the Weierstrass form~\cite{BMX06}.}. Though this approach did not bring a dramatic efficiency as that of Barreto et al. for Tate pairing computation, but it can be applied for computing both Weil and Tate pairings on any pairing-friendly elliptic curve.

Recently, Boxall et al. \cite{BELL10} presented a variant of Miller's algorithm due to a variant of Miller's formulas. Similar to the approach of Blake et al., their algorithm can also be applied on general elliptic curves.

In this paper we extend the Blake-Murty-Xu's method and show how to eliminate all of vertical lines in Miller's algorithm. Our algorithm is \emph{generically} faster than the original Miller's algorithm, and its refinements~\cite{BMX06, LHC07} for all pairing-friendly curves with any embedding degree. As previous refinements, our algorithm does not eliminate denominators, but it  improves the performance for both Weil and Tate pairings computation on {\em general} pairing-based elliptic curves. Our algorithm is of particular interest to compute the Ate-style pairings on elliptic curves with small embedding degrees $k$, and in situations where denominators elimination using a twist is not possible (for example on curves with embedding degree $k$ not of the form $2^i3^j$, where $i \ge 1$, $j \ge 0$). 

We also study, in this paper, a modification of our algorithm which can eliminate denominators when computing Tate pairing on elliptic curves with even embedding degree. The efficiency of this modified algorithm can thus be comparable to that of Barreto et al.~\cite{BLS03}.

The rest of the paper is organized as follows. We briefly recall definitions of the Weil/Tate pairings, Miller's algorithm and the Blake-Murty-Xu's method in Section 2. Section 3 presents our improvements to the original Miller's algorithm for general elliptic curves. Section 4 analyzes theoretically the efficiency of our algorithm and compares it with previous improvements. Section 5 will discuss a modification without denominators applicable when the embedding degree $k$ is even. Section 6 will give some experimental results. The conclusion and open problems will be given in Section 7.

\section{Preliminaries}
In this section, we give a brief summary of several mathematical background and the definitions of the Weil/Tate pairings. We review then Miller's algorithm for Weil/Tate pairing computation. Finally, we briefly recall the Blake-Murty-Xu's method for reducing vertical lines in Miller's algorithm. 

\subsection{Divisors and Bilinear Pairings}
Let $K = \mathbb{F}_q = \mathbb{F}_{p^m}$ be a finite field of $p$ characteristic with $q$ elements and $p > 3$ must be a prime number. An elliptic curve $E$ defined over $K$ in short Weierstrass form is the set of solutions $(x, y)$ to the following equation: 

$$E : y^2 = x^3 + ax + b,$$

together with an extra point ${\cal O}$ which is called the \textit{point at infinity} of $E$. Where $a, b \in \F_q$ such that the discriminant $\Delta = 4a^3 + 27b^2$ is non-zero.

A \textit{divisor} is an element of the free abelian group $Div(E)$ generated by the points of $E$. Given a divisor $D = \sum_{P \in E} n_P(P)$, where $n_P \in \mathbb{Z}$ and only a finite number of the integers $n_P$ are nonzero, the \textit{degree} of $D$, denoted $\text{deg}D$, is the integer $\sum_{P \in E} n_P$, and the order of $D$ at the point $P$, denoted $\text{ord}_P(D)$, is the integer $n_P$. The \textit{support} of $D$ is the set of point $P$ such that $n_P \ne 0$.

Let $f \in K(E)$ be a nonzero rational function. The divisor of $f$ is indeed a finite formal sum
$$div(f) = \sum_{P \in E} \text{ord}_P(f)(P),$$
where $\text{ord}_P(f)$ is the order of the zero or pole of $f$ at $P$, that is, $\text{ord}_P(f) > 0$ if $P$ is a zero of $f$, $\text{ord}_P(f) < 0$ if $P$ is a pole of $f$, and $\text{ord}_P(f) = 0$ otherwise. It follows from the definition that $div(fg) = div(f) + div(g)$ and $div(f/g) = div(f) - div(g)$ for any two nonzero rational functions $f$ and $g$ defined on $E$. 

Let $Div^0(E)$ be the subgroup of $Div(E)$ consisting of divisors of degree $0$. It turns out that $div(f)$ has the degree $0$ (i.e. $div(f) \in Div^0(E)$) and is called \textit{principal divisor}. A divisor $D$ is called \textit{principal} if $D = div(f)$ for some function $f$. It is known that a divisor $D = \sum_{P \in E} n_P(P)$ is principal if and only if the degree of $D$ is zero and $\sum_{P \in E} n_P P = \cal{O}$. Two divisors $D_1$ and $D_2$ are equivalent on $Div^0(E)$, denoted $D_1 \sim D_2$, if and only if their difference $D_1 - D_2$ is a principal divisor\footnote{We refer the readers to~\cite{Kob98} for more details about divisors and rational functions.}.

The key to the definition of pairings is the evaluation of rational functions in divisors. For any function $f$ and any divisor $D= \sum_{P \in E} n_P(P)$ of degree $0$, we define $f(D) = \prod_{P \in E}f(P)^{n_P}$. Let $r$ be an integer co-prime to the characteristic $p$ of $E$, and $P, Q \in E[r]$, where $E[r]$ is the set of points of order $r$. Let $D_P, D_Q \in Div^0(E)$ be two divisors which are equivalent to $(P) - (\mathcal{O})$ and $(Q) - (\mathcal{O})$, respectively and such that $D_P$ and $D_Q$ have disjoint supports. As $rD_P$ and $rD_Q$ are principal, and hence there exist functions $f_P, f_Q$ such that $div(f_P) = rD_P$ and $div(f_Q) = rD_Q$. Then the Weil pairing $\omega : E[r] \times E[r] \mapsto \mathbb{F}_{q^k}$ is defined as
$$\omega(P, Q) = \frac{f_P(D_Q)}{f_Q(D_P)},$$
where $k$ is the embedding degree of $E(\mathbb{F}_q)$.

The Tate pairing is also defined based on $f_P(D_Q)$. Let $P \in E(\mathbb{F}_q)[r]$ and $Q \in E(\mathbb{F}_{q^k})[r]$ be linearly independent points. Then the (reduced) Tate pairing $\tau : E(\mathbb{F}_q)[r] \times E(\mathbb{F}_{q^k})[r] \mapsto \mathbb{F}_{q^k}^*$ is defined as
$$\tau(P, Q) = f_P(D_Q)^{\frac{q^k - 1}{r}}.$$

The Weil/Tate pairings satisfy the properties: \textit{bilinearity}, \textit{non-degeneracy} and \textit{compatibility with isogenies}. 
\paragraph*{{\bf The twist of a curve}.}

Let $d$ be a factor of $k$, an elliptic curve $E'$ over $\mathbb{F}_{q^{k/d}}$ is called a twist of degree $d$ of $E$ if there exists an isomorphism $\psi : E' \mapsto E$ defined over $\mathbb{F}_{q^{d}}$.

A twist of $E$ is given by $E': y^2 = x^3 + a\beta^4 x + b\beta^6$ for some $\beta \in \mathbb{F}_{q^k}$. The isomorphism between $E'$ and $E$ is $\psi : E' \mapsto E : (x', y') \mapsto (x'/\beta^2, y'/\beta^3)$.

\subsection{The Miller's Algorithm}

The pairings over (hyper)elliptic curves are computed using the algorithm proposed by Miller~\cite{Mil04}. The main part of the Miller's algorithm is constructing the rational function $f_{r,P}$ and evaluating $f_{r,P}(Q)$ with $div(f_{r,P}) = r(P) - (rP) - [r - 1]({\cal O})$ for divisors $P$ and $Q$. Let $G_{iP,jP}$ be a rational function with 
          $$div(G_{iP,jP}) = (iP) + (jP) - ([i+j]P) - (\mathcal{O})$$

Miller's algorithm is based on the following relation describing the
so-called \emph{Miller's formula}, which is proved by considering divisors.

          $$f_{i+j,P} = f_{i,P} f_{j,P} G_{iP,jP}.$$

Let $L_{iP, jP}$ be equation of the line passing through $iP$ and $jP$ (or the equation of the tangent line to the curve if $i = j$). Let $V_{iP + jP}$ be equation of the vertical line passing through $(iP + jP)$ and $-(iP + jP)$. In the case of elliptic curves, then
$$div(L_{iP, jP}) = (iP) + (jP) + (-[i + j]P) - 3(\mathcal{O}), \text{ and }$$
$$div(V_{[i + j]P}) = ([i + j]P) + (-[i + j]P) - 2(\mathcal{O}).$$
Thus, we have
$$G_{iP, jP} = \frac{L_{iP, jP}}{V_{(i + j)P}}.$$

We say that in the case of elliptic curves, $G_{iP,jP}$ is the line passing through the points $iP$ and $jP$ divided by the vertical line passing through the point $[i + j]P$.

Notice that $div(f_0) = div(f_1) = 0$, so that $f_0 = f_1 = 1$. Let the binary representation of $r$ be $r = \sum_{i = 0}^t b_i 2^i$, where $b_i \in \{0, 1\}$. Using the double-and-add method, Miller's algorithm is described as in {\bf Algorithm~\ref{algo:Miller}}.

\begin{algorithm}[H]
   %\SetLine
   \KwIn{  $r = \sum_{i = 0}^t b_i 2^i$ with $b_i \in \{0, 1\}$, $P, Q \in E[r]$;}
   \KwOut{ $f = f_r(Q)$;}
  \BlankLine
   { 
   $T \leftarrow P$, $f \leftarrow 1$\;
 \For{$i = t - 1$  \emph{\KwTo}  $0$} 
 {    
    \nl $f \leftarrow {f}^2 \frac{L_{T, T}(Q)}{V_{2T}(Q)}$, $T \leftarrow 2T$ \; \label{Miller:line1}
   \If{$b_i = 1$}   {  
    \nl ${f} \leftarrow {f} \frac{L_{T, P}(Q)}{V_{T + P}(Q)}$, $T \leftarrow T + P$ \;\label{Miller:line2}
   }  }
    \KwRet{ $f$ } }
   \caption{Miller's Algorithm $(P, Q, r)$}
   \label{algo:Miller}

\end{algorithm}

\subsection{Blake-Murty-Xu's method}
Blake et al. achieved three refinements to Miller's algorithm in~\cite{BMX06} thanks to the following observation:
\begin{lemma}[Lemma 1,~\cite{BMX06}]
If the line $h(x, y) = 0$ intersects with $E$ at points $P = (a, b)$, $Q = (c, d)$ and $-(P + Q)$ with $P + Q = (\alpha, \beta)$, then
$$h(x, y)\bar{h}(x, y) = -(x-a)(x-c)(x-\alpha).$$ 
\label{lem1}
\end{lemma}
The proof was proven in ~\cite{BMX06}. From this lemma, they gave the following equation:
\begin{equation}
\frac{L_{T, T}(Q)}{V_{T}^2(Q)V_{2T}(Q)} = -\frac{1}{L_{T, T}(-Q)}, 
\label{eq1}
\end{equation} 
where $L_{T, T}(-Q) = \bar{L}_{T, T}(Q)$.

The way of Blake et al. to apply this observation to Miller's algorithm is simple. They delay a vertical line ($V_{2T}$ in a doubling step or $V_{T + P}$ in an addition step) in the denominator at each iteration for the next iteration. Two vertical lines can thus be eliminated as we can see in Eq.~(\ref{eq1}). When most bits in the binary representation of $r = \sum_{i = 0}^t b_i 2^i$ are $1$, the author described the algorithm as in {\bf Algorithm~\ref{algo:BMX2}}.

\begin{algorithm}[H]
   %\SetLine
    
   \KwIn{  $r = \sum_{i = 0}^t b_i 2^i$ with $b_i \in \{0, 1\}$, $P, Q \in E$ where $P$ has order $r$;}
   \KwOut{ $f = f_r(Q)$;}
    \BlankLine
   { 
    \eIf{$b_{t - 1} = 0$} {
      $f \leftarrow L_{P, P}$; $T \leftarrow 2P$ \;
    }
    {
      $f \leftarrow \frac{L_{P, P}(Q)L_{2P, P}(Q)}{V_{2P}(Q)}$; $T \leftarrow 3P$ \;
    }
   
 \For{$i = t - 2$  \emph{\KwTo}  $0$} 
 {    
   \eIf{$b_i = 0$}   {  
	$f \leftarrow {f}^2 \frac{V_{2T}(Q)}{L_{T, T}(-Q)}$; $T \leftarrow 2T$ \;      
   }
   {  
    ${f} \leftarrow f^2 \frac{L_{2T, P}(Q)}{L_{T, T}(-Q)}$; $T \leftarrow 2T + P$ \;
    }  
  }
    \KwRet{ $f$ }}   
    
\caption{Improved Miller's Algorithm of Blake et al. (Algorithm~4 in~\cite{BMX06})}
\label{algo:BMX2}

\end{algorithm}

The Blake et al.'s algorithm eliminated all (two) vertical lines if the bit $b_i$ is $1$. When the bit $b_i$ is zero, the computation cost is the same as in Miller's algorithm. Thus, the rest number of vertical lines that has to be computed is the number of bits $0$ in binary representation of $r$, or $l = t - H(r)$, where $H(r)$ is the Hamming weight of $r$ and $t$ is the number of bits of $r$. This algorithm works well when the Hamming weight $H(r)$ is high. When $H(r)$ is low, the authors also presented a refinement to Miller's algorithm in radix 4 ($r = \sum_{i = 0}^{t/2} q_i 4^i$, with $q_i \in \{0, 1, 2, 3\}$). This refinement can save at most all of two vertical lines when a pair of consecutive bits is ``$00$'' ($q_i = 0$). However, there still exists one vertical line in the case of $q = \{1, 2\}$, and two vertical lines if $q_i = 3$. 

Then, Liu et al.~\cite{LHC07} presented a further refinement to the Miller's algorithm by using the Blake-Murty-Xu's method. Their improvement requires an additional algorithm for segmenting the binary representation of the order of subgroups $r$ into 7 cases : $(00)^{i}$, $(00)^{i}0$, $(1)^{i}$, $(01)^{i}$, $0(1)^{i}$, $(1)^{i}0$ and $0(1)^{i}0$. Though the algorithm is complex, but it allowed to reduce more lines than Blake-Murty-Xu's algorithm. 

\section{Our Improvement on Miller's Algorithm}

Firstly, we present two following lemmas that are the same as lemma 1 and lemma 2 of \cite{BMX06} if we replace $L_{iP, jP}(-Q)$ by $L_{-iP, -jP}(Q)$. However, unlike the proof of lemma~1 in \cite{BMX06} that makes use of a geometrical observation, our lemma is achieved by calculating divisors.
\begin{lemma}
\label{lem2}
Let $L_{iP, jP}$ be equation of the line passing through $iP$ and $jP$, $L_{-iP, -jP}$ be equation of the line passing through $-iP$ and $-jP$, and let $V_{iP + jP}$ be equation of the vertical line passing through $(iP + jP)$ and $-(iP + jP)$. Then,
\begin{equation}
L_{iP, jP}(Q)L_{-iP, -jP}(Q) = V_{iP}(Q)V_{jP}(Q)V_{[i+j]P}(Q). 
\label{eq2}
\end{equation}
 
\end{lemma}

\begin{proof} By calculating divisors, it is straightforward to see that:
\begin{align*}
div(&L_{iP, jP}(Q)L_{-iP, -jP}(Q)) \\&= div(L_{iP, jP}(Q)) + div(L_{-iP, -jP}(Q)) \\
				  &= (iP) + (jP) + (-[i + j]P) - 3(\mathcal{O}) + \\& \qquad + (-iP) + (-jP) + ([i + j]P) - 3(\mathcal{O})\\
				  &= (iP) + (-iP) + (jP) + (-jP) + \\& \qquad+ (-[i + j]P) + ([i + j]P) - 6(\mathcal{O})
\end{align*}
\begin{align*}
div(&V_{iP}(Q)V_{jP}(Q)V_{[i+j]P}(Q)) \\&= div(V_{iP}(Q)) + div(V_{jP}(Q)) + div(V_{[i+j]P}(Q))\\
				  &= (iP) + (-iP) -2(\mathcal{O}) + (jP) + (-jP) -\\ & \qquad - 2(\mathcal{O}) + ([i + j]P) + (-[i + j]P) - 2(\mathcal{O})\\
				  &= (iP) + (-iP) + (jP) + (-jP) +\\ & \qquad + (-[i + j]P)  + ([i + j]P) - 6(\mathcal{O})
\end{align*}
Thus, Eq.~(\ref{eq2}) is hold.
\end{proof}

\begin{lemma}
\begin{equation}
\frac{L_{T, T}(Q)}{V_{T}^2(Q)V_{2T}(Q)} = \frac{1}{L_{-T, -T}(Q)}. 
\label{eq3}
\end{equation} 
\end{lemma}

This lemma is easy to be proven using the above lemma.

In what follows, we use the notation $L_{T, T}$ replacing for $L_{T, T}(Q)$, and the notation $V_T$ replacing for $V_T(Q)$. 

\subsection{Blake-Murty-Xu's Refinement}
Suppose that the order of subgroup $r$ have the binary representation $r = \sum_{i = 0}^t b_i 2^i$. The rational function $f_r$ can be displayed as follows:

\begin{equation}
\label{eq6}
 f_r = f_1^r \prod_{i = t}^{1} \left(\frac{L_{\lfloor \frac{r}{2^i}\rfloor P, \lfloor \frac{r}{2^i} \rfloor P}}{V_{2\lfloor \frac{r}{2^{i}}\rfloor P}}\frac{L_{2\lfloor \frac{r}{2^i}\rfloor P, b_{i - 1}P}}{V_{\lfloor \frac{r}{2^{i - 1}}\rfloor P}} \right)^{2^{i - 1}}
\end{equation}

In this formula, if $b_{i - 1} = 0$, then $L_{2\lfloor \frac{r}{2^i}\rfloor P, b_{i - 1}P} = L_{2\lfloor \frac{r}{2^i}\rfloor P, \cal{O}} = V_{\lfloor \frac{r}{2^{i - 1}}\rfloor P}$. We also assume that $V_{rP} = V_{\cal{O}} = 1$.

In the case when most bits in the binary representation of $r$ are $1$, Blake et al. given the following computations:
\begin{eqnarray}
\label{eq7}
 f_r &=& f_1^r \left(\frac{L_{P, P}L_{2P, b_{t-1}P}}{V_{2P}}\right)^{2^{t - 1}} \times \nonumber \\ &\qquad& \prod_{i = t - 1}^{1} \left(\frac{L_{\lfloor \frac{r}{2^i}\rfloor P, \lfloor \frac{r}{2^i} \rfloor P}}{V_{\lfloor \frac{r}{2^{i}}\rfloor P}^2}\frac{L_{2\lfloor \frac{r}{2^i}\rfloor P, b_{i - 1}P}}{V_{2\lfloor \frac{r}{2^i}\rfloor P}} \right)^{2^{i - 1}} \nonumber \\ 
    &=& f_1^r \left(\frac{L_{P, P}L_{2P, b_{t-1}P}}{V_{2P}}\right)^{2^{t - 1}} \times \nonumber \\ &\qquad& \prod_{i = t - 1}^{1} \left(\frac{L_{2\lfloor \frac{r}{2^i}\rfloor P, b_{i - 1}P}}{L_{-\lfloor \frac{r}{2^i}\rfloor P, -\lfloor \frac{r}{2^i} \rfloor P}} \right)^{2^{i - 1}}
\end{eqnarray}

In this refinement, Blake et al. always make a delay of one vertical line for the next step for the purpose of applying Lemma~\ref{lem1}. This trick runs well with the bit $b_i = 1$. However, in the case of $b_i$ is $0$ (there is only one vertical line dealt in this step), they added a vertical line into the current step. The number of vertical lines remaining to be calculated is thus equal that of bits $0$.

\subsection{Our Refinement}
From the observations in~\cite{BMX06, LHC07} and by combining with the Eisentr{\"a}ger, Lauter and Montgomery's trick~\cite{ELM03}, we present a modification to Miller's algorithm that can eliminate all vertical lines. 

The function $f_r$ in Eq.~(\ref{eq6}) is the product to start from term $t$ to term $1$. Let $f^{(i)}_r$ be the value of $f_r$ at the term $i$, we re-define the function $f_r$ as follows:

\begin{equation}
  f^{(i)}_r = 
     \begin{cases}
        f_1^r & \qquad \text{if $\quad i > t$} \\
        f^{(i + 1)}_r (g^{(i)})^{2^{i - 1}} & \qquad \text{if $\quad 1 < i \le t$}  \\
	f_r & \qquad \text{if $\quad i = 1$,}
     \end{cases}
\end{equation}

where the function $g^{(i)}$ is defined as follows:
\begin{equation}
\label{func_g}
  g^{(i)} = 
          \begin{cases}
            \frac{L_{\lfloor \frac{r}{2^i}\rfloor P, \lfloor \frac{r}{2^i} \rfloor P} \cdot L_{2\lfloor \frac{r}{2^i}\rfloor P, b_{i - 1}P}}{V_{2\lfloor \frac{r}{2^i}\rfloor P}} & \quad \text{if $m_{i + 1} = 0$} \\
            \frac{L^{b_{i - 1}}_{2\lfloor \frac{r}{2^i}\rfloor P, P}}{L_{-\lfloor \frac{r}{2^i}\rfloor P, -\lfloor \frac{r}{2^i} \rfloor P}} & \quad \text{if $m_{i + 1} = 1$,}
          \end{cases}
\end{equation}

If $b_{i - 1} = 0$, then $L_{2\lfloor \frac{r}{2^i}\rfloor P, b_{i - 1}P} = L_{2\lfloor \frac{r}{2^i}\rfloor P, \cal{O}} = V_{2\lfloor \frac{r}{2^i}\rfloor P}$. In this case, there is no more vertical lines in Eq.~\ref{func_g}. Otherwise, $b_{i - 1} = 1$ we will show in the following subsection how to apply the Eisentr{\"a}ger-Lauter-Montgomery's trick \cite{ELM03} to eliminate the vertical line $V_{2\lfloor \frac{r}{2^i}\rfloor P}$ in the equation $\frac{L_{\lfloor \frac{r}{2^i}\rfloor P, \lfloor \frac{r}{2^i} \rfloor P} \cdot L_{2\lfloor \frac{r}{2^i}\rfloor P, P}}{V_{2\lfloor \frac{r}{2^i}\rfloor P}}$.

In the above equation, $m_i$ is defined as follows:
\begin{equation}
   m_i = 
         \begin{cases}
            0 & \qquad \text{if $\quad i > t$} \\
            \lnot m_{i + 1} \text{ or } b_{i-1}  & \qquad \text{if $\quad  1 \le i \le t$.} \\
          \end{cases}
\end{equation}

Unlike as the Blake-Murty-Xu's refinement, we accept that maybe there is not any line delayed in some steps. If $m_i = 1$, there is a line delayed for the next step and otherwise. For $1 \le i \le t$, $m_i$ become $0$ if and only $m_{i + 1} = 1$ and the bit $b_{i - 1} = 0$.

\subsection{Our Algorithm}

We use a memory variable $m$ to note that whether there is still a vertical line delayed in the current step or not. At each step, we will apply Eq.~(\ref{eq3}) if $m = 1$. Without loss of generality, we assume that $f_1 = 1$ and $V_{rP} = 1$.

The algorithm is described by the pseudocode as in {\bf Algorithm~\ref{our:algoMiller}}.

\begin{algorithm}
   
   %\SetLine
   \KwIn{  $r = \sum_{i = 0}^t b_i 2^i$, $b_i \in \{0, 1\}$.}
   \KwOut{ $f$}
  \BlankLine
   { 
   $T \leftarrow P$, $f \leftarrow 1$, $m \leftarrow 0$\;
 \For{$i = t - 1$  \emph{\KwTo}  $0$} 
 {  
    \nl \If{$b_i = 0$ and $m = 0$} { \label{eco1}
      $f \leftarrow f^2 \cdot L_{T,T}$; $T \leftarrow 2T$;	$m \leftarrow 1$\; 
    }

    \nl \If{$b_i = 0$ and $m = 1$} {\label{eco2}
      $f \leftarrow f^2 \cdot \frac{1}{L_{-T,-T}}$; $T \leftarrow 2T$; $m \leftarrow 0$ \; 
    }

    \nl \If{$b_i = 1$ and $m = 1$} { \label{eco3}
      $f \leftarrow f^2 \cdot \frac{L_{2T,P}}{L_{-T,-T}}$;  $T \leftarrow 2T + P$;   $m \leftarrow 1$\;
    }

    \nl \If{$b_i = 1$ and $m = 0$} { \label{parabola}
      $f \leftarrow f^2 \cdot \frac{L_{T,T} \cdot L_{2T,P}}{V_{2T}}$; $T \leftarrow 2T + P$; $m \leftarrow 1$\;
    }

   }
%  }	
  
  \KwRet{ $f$} 
}

\caption{Improved Refinement of Miller's Algorithm for any Pairing-Friendly Elliptic Curve}
\label{our:algoMiller}

 \end{algorithm}

%The algorithm always delay a vertical line for next step except the step~(\ref{eco2}) (the case $b_i = 0$ and $m = 1$).
\begin{remark}
As the original Miller's algorithm, our algorithm cannot avoid divisions needed to update $f$. But we can reduce them easily to one inversion at the end of the addition chain (for the cost of one squaring in addition at the each step of the algorithm).
\end{remark}

We can see that the algorithm eliminated all of vertical lines except the case of line~{\bf 4} of the Figure~\ref{our:algoMiller}. Now, we will show how to use the Eisentr{\"a}ger, Lauter and Montgomery's trick in~\cite{ELM03} to replace the quotient by a parabola equation.

\paragraph*{\textbf{Eisentr{\"a}ger-Lauter-Montgomery's trick}} In~\cite{ELM03}, the authors gave significant and useful application for computing $f_{(2i + j),P}$ directly from $f_{i,P}$ and $f_{(i + j),P}$ instead of traditional double-and-add method. They constructed a parabola, whose formula can be used to replace $\frac{L_{iP, jP} L_{[i+j]P, iP}}{V_{[i+j]P}}$, through the points $iP, iP, jP, -2iP - jP$ as follows.

Let $iP + jP = (x_3, y_3)$ and $2iP + jP = (x_4, y_4)$. Then, 
\begin{eqnarray}
&\;&\frac{L_{iP, jP} L_{[i+j]P, iP}}{V_{[i+j]P}} \nonumber \\& = & \frac{(y + y_3 - \lambda_1(x - x_3))(y - y_3 - \lambda_2(x - x_3))}{x-x_3}
\label{eq4} 
\end{eqnarray}

We simplify the right half of Eq.~(\ref{eq4}) by expanding it in powers of $x-x_3$ and obtaining the following parabola.
\begin{align*}
 \frac{y^2 - y_3^2}{x - x_3} &- \lambda_1(y - y_3)  - \lambda_2(y + y_3) + \lambda_1\lambda_2(x - x_3) \\
    &= x^2 + x_3 x + x_3^2 + a_4 + \lambda_1\lambda_2(x - x_3)- \\& \qquad - \lambda_1(y - y_3)  - \lambda_2(y + y_3) \\
    &= x^2 + (x_3 + \lambda_1\lambda_2)x - (\lambda_1 + \lambda_2)y + constant \\
    &= (x - x_1)(x + x_1 + x_3 + \lambda_1\lambda_2) - \\& \qquad - (\lambda_1 + \lambda_2)(y - y_1).
\end{align*}

Clearly, this substitution parabola needs less effort to evaluate at a point than to evaluate the quotient $\frac{L_{iP, jP} L_{[i+j]P, iP}}{V_{[i+j]P}}$ at that point. Additionally, the parabola does not reference $y_3$, so we can save one multiplication for calculating $2T + P$ by using the double-add trick of Eisentr{\"a}ger et al.

Now, we apply the Eisentr{\"a}ger, Lauter and Montgomery's method to construct a parabola replacing for $\frac{L_{T,T} L_{2T,P}}{V_{2T}}$. Similarly, let $2T = (x_3, y_3)$, then
\begin{eqnarray}
&&P_{T, P}(x, y) = \frac{L_{T,T} L_{2T,P}}{V_{2T}} = \nonumber \\ &=& \frac{(y + y_3 - \lambda_1(x - x_3))(y - y_3 - \lambda_2(x - x_3))}{x-x_3} \nonumber \\ &=& (x - x_1)(x + x_1 + x_3 + \lambda_1\lambda_2) - \nonumber \\ & \qquad & - (\lambda_1 + \lambda_2)(y - y_1),
\label{eq5} 
\end{eqnarray}
where $\lambda_1$ is the slope of the line passing through $T$ twice and $-2T$, $\lambda_2$ is the slope of the line passing through $2T$, $P$ and $-2T - P$. The quotient also has zeros at $T$ twice (i.e., tangent at $T$), $P$ and $-2T - P$ and a pole of order $4$ at $\mathcal O$. By simplifying as above, we obtain a substitution parabola.

The table~\ref{tab:Cmp} shows that our algorithm is more efficient than the classical Miller's algorithm as we save a product in
the full extension field at each doubling and each addition step. The following subsection discusses all this in more detail. In Section \ref{sec:nongen} we describe a version without denominators that can be applied for computing Tate pairing on elliptic curves with even embedding degree.

\section{Efficiency comparison}
In this section we will give a performance analysis of our algorithm and make a comparison among the original Miller's algorithm~\cite{Mil04}, the Blake-Murty-Xu's refinements~\cite{BMX06}, the Barreto et al.'s algorithm for computing the Tate pairing on curves with even embedding degrees~\cite{BLS03} and Lin et al.'s algorithm for computing the pairings on curves with the embedding degree $k = 9$~\cite{LZZW08}. 

One can consider that the cost of the algorithms for pairing computation consists of three parts: the cost of updating the function $f$, the cost of updating the point $T$ and the cost of evaluating rational functions at some point $Q$. Without special treatment, we consider that the cost of updating $T$ and the cost of evaluating rational functions $L_{T, T}$, $L_{2T, P}$ at the point $Q$ are the same for all algorithms (the cost of evaluating $L_{-T, -T}$ at a point is the same that of evaluating $L_{T, T}$ at that point). Besides, the most costly operations in pairing computations are those that take place in the full extension field $\F_{q^k}$. At high levels of security (i.e. $k$ large), the complexity of operations in $\F_{q^k}$ dominates the complexity of the operations that occur in the lower degree subfields. 

Because of the above reasons, we only focus on the cost of updating the function $f$ which is generally executed on the full extension field $\mathbb{F}_{q^k}$.

Let {\bf M}, {\bf S} and {\bf I} denote the cost of one full extension field multiplication, one full extension field squaring and one full extension field inversion respectively for updating $f$. In following analysis, the ratio of one full extension field squaring to one full extension field multiplication is set to {\bf S = 0.8M}, a commonly used value in the literature (see \cite{BHLM01}). The cost of one full extension field division that consists of one full extension field inversion {\bf I} and one full extension field multiplication {\bf M}, is generally several times more than one full extension field multiplication \cite{BHLM01, LMN10}. To avoid this, we manipulate the numerators and denominators \emph{separately}, and perform one division at the very end of the algorithm (for the cost of one full extension field squaring {\bf S} in addition for each bit treated). 

In~\cite{BLS03}, Barreto et al. pointed out that when the embedding degree $k$ is even, denominators can be totally eliminated during Tate pairing computation. The authors observed that the point $Q$ can be chosen so that its $x$-coordinate lie in a proper subfield, the valuation of the vertical line $V_{T + P} = x_Q - x_{T + P}$ would be in a proper subfield of $\mathbb{F}_{q^k}$. Thus the denominator would become $1$ when the final exponentiation is performed. Similarly, Lin et al.~\cite{LHC07} proposed an algorithm that can eliminate denominators during Tate pairing computation on curves with the embedding degrees $k = 3^i$ that can employ a cubic twist. On the other hand, their algorithm needs one full extension field multiplication compared to that of Barreto et al~\cite{BLS03}.

Table~\ref{tab:Cmp} gives a comparison of the cost of updating $f$ between our algorithm ({\bf Algorithm~\ref{our:algoMiller}}) with that in Miller's algorithm~\cite{Mil04} ({\bf Algorithm~\ref{algo:Miller}}), Blake-Murty-Xu's algorithms in~\cite{BMX06} ({\bf Algorithm~\ref{algo:BMX2}}, and {\bf Algorithm~\ref{algo:BMX4}} described in \ref{appd1}), Barreto et al.'s algorithm~\cite{BLS03} and Lin et al's algorithm~\cite{LZZW08}. 

\begin{table*}
	\centering
		\begin{tabular}{|l|c|c|}
		\hline
			& Doubling & Doubling and Addition \\
		\hline
		{\bf Algorithm~\ref{algo:Miller}} & 2{\bf S} + 2{\bf M} & 2{\bf S} + 4{\bf M}\\
		(Miller's algorithm~\cite{Mil04})&	= 3.6{\bf M} &	= 5.6{\bf M}	\\
		\hline
		{\bf Algorithm~\ref{algo:BMX2}} & 2{\bf S} + 2{\bf M} & 2{\bf S} + 2{\bf M}\\
		(Algorithm~4 in \cite{BMX06})	&	= 3.6{\bf M} &	= 3.6{\bf M} \\
		\hline
		{\bf Algorithm~\ref{algo:BMX4}} & 2{\bf S} + 1{\bf M} & 2{\bf S} + 3{\bf M}\\
		(Algorithm~3 in \cite{BMX06})	&	= 2.6{\bf M} &	= 4.6{\bf M} \\
		\hline
		\multirow{2}{*}{Barreto et al.'s algorithm~\cite{BLS03}} & 1{\bf S} + 1{\bf M} & 1{\bf S} + 2{\bf M}\\
							& = 1.8{\bf M}	& = 2.8{\bf M} \\
		\hline
		\multirow{2}{*}{Lin et al.'s algorithm~\cite{LZZW08}} &  1{\bf S} + 2{\bf M} & 1{\bf S} + 4{\bf M}\\
						 &	= 2.8{\bf M}	&	= 4.8{\bf M} \\
		\hline
		\multirow{2}{*}{{\bf Algorithm~\ref{our:algoMiller}}} & 2{\bf S} + 1{\bf M} & 2{\bf S} + 2{\bf M} = 3.6{\bf M}(line {\bf 3}) \\
		& = 2.6{\bf M} & 2{\bf S} + 1{\bf M} = 2.6{\bf M} (line {\bf 4})\\
									
		\hline
\end{tabular}
	\caption{Comparison of the cost of updating $f$ of Algorithms. ``Doubling'' is when algorithms deal with the bit ``$b_i = 0$'' and ``Doubling and Addition'' is when algorithms deal with the bit ``$b_i = 1$''.}
	\label{tab:Cmp}
\end{table*}

From Table~\ref{tab:Cmp}, for the \emph{generic} case we can see that {\bf Algorithm~\ref{our:algoMiller}} saves one full extension field multiplication when the bit $b_i = 0$ compared with {\bf Algorithm~\ref{algo:Miller}} and {\bf Algorithm~\ref{algo:BMX2}}. It has the same cost as {\bf Algorithm~\ref{algo:BMX4}}. When the bit $b_i = 1$, {\bf Algorithm~\ref{our:algoMiller}} has the same cost as {\bf Algorithm~\ref{algo:BMX2}} but saves one and two full extension field multiplications in comparison to {\bf Algorithm~\ref{algo:BMX4}} and {\bf Algorithm~\ref{algo:Miller}}, respectively. In total, our algorithm saves $\log(r) + H(r)$, $\log(r) - H(r)$ and $H(r)$ full extension field multiplications compared with the original Miller's algorithm, {\bf Algorithm~\ref{algo:BMX2}} and {\bf Algorithm~\ref{algo:BMX4}}, respectively. Here, $\log(r)$ and $H(r)$ denote the length in bits of the elliptic curve group order and the Hamming weight of the group order $r$, respectively. 

When the embedding degree $k$ gets large, the the complexity of the operations occurring in the full extension field $\F_{q^k}$ dominates the complexity of those operations occurring in $\F_q$, then, our algorithm is faster about 25\% than the original Miller's algorithm.

Our algorithm is also better than that of Lin et al.~\cite{LZZW08} in all case. It tradeoffs one ${\bf S-M}$ in the doubling step and only requires 2{\bf S} + 2{\bf M} instead of 1{\bf S} + 4{\bf M} for doubling and addition step as in their algorithm. %As pointed out in~\cite{LZZW08}, at the AES $128$ security level, the performance of Lin et al.'s algorithm on curves with $k = 9$ is comparable to the Barreto-Naehrig curves with $k = 12$ \cite{BN06}. As already mentioned, our approach is generic and thus the algorithm in Figure~\ref{our:algoMiller} can be used to replace Lin et al.'s algorithm and allows a faster computation. 

In comparison with Barreto et al.'s algorithm \cite{BLS03}, our algorithm takes one more full extension field squaring for each bit. However, as already mentioned, our approach is generic and it can be applied on \emph{any} (pairing-friendly) elliptic curve.

The next section present a modification of our algorithm which can be used for computing pairings on elliptic curves with even embedding degree. We show that the efficiency of the modified algorithm is comparable to Barreto et al.'s algorithm. 

\section{A modification for elliptic curves with even embedding degree}
\label{sec:nongen}

Actual implementations are adapted to twisted elliptic curves, thus the Miller's algorithm can be implemented more efficiently. Indeed, as pointed out in~\cite{BLS03} such curves admit an even twist which allows to eliminate denominators and all irrelevant terms in  the subfield of $\F_{q^k}$. In the case of a cubic twist, denominator elimination is also possible \cite{LZZW08}. Another advantage of embedding degrees of the form $2^i3^j$, where $i \ge 1$, $j \ge 0$ is that the corresponding extensions of $\F$ can be written as composite extensions of degree $2$ or $3$, which allows faster basic arithmetic operations \cite{KM05}.

In this subsection we construct a variant of {\bf Algorithm~\ref{our:algoMiller}} in the case of $k$ even.

Let $v = (a + ib)$ be a representation of an element of $\F_{q^k}$, where $a$, $b \in \F_{q^{k/2}}$, and $i$ is a quadratic non-residue and $\delta = i^2$. The conjugate of $v$ over $\F_{q^{k/2}}$ is given by $\bar{v} = \overline{(a + ib)} = a - ib$. It follows that, if $v\neq 0$, then
\begin{equation*}
\frac{1}{v}=\frac{\overline{v}}{a^2-\delta b^2}
\end{equation*}
where $a^2-\delta b^2 \in \mathbb{F}_{q^{k/2}}$. Thus, in a situation where elements of $\F_{q^{k/2}}$ can be ignored, $\frac{1}{v}$ can be replaced by $\overline{v}$, thereby saving an inversion in $\F_{q^k}$ \cite{Sco06}.

We exploit this fact in the following modification of the algorithm, where
we replace the denominator $L_{-T, -T}$ by its conjugate $\overline{L_{-T, -T}}$.

The new algorithm works as follows:% given in {\bf Algorithm~\ref{our:algoMiller1}}

%\begin{figure}[htbp]
 
\begin{algorithm}%[H]
   
   %\SetLine
   \KwIn{  $r = \sum_{i = 0}^t b_i 2^i$, $b_i \in \{0, 1\}$.}
   \KwOut{ $f$}
  \BlankLine
   { 
   $T \leftarrow P$, $f \leftarrow 1$, $m \leftarrow 0$\;
 \For{$i = t - 1$  \emph{\KwTo}  $0$} 
 {  
    \nl \If{$b_i = 0$ and $m = 0$} { \label{eco12}
      $f \leftarrow f^2 L_{T,T}$; $T \leftarrow 2T$;	$m \leftarrow 1$\; 
    }

    \nl \If{$b_i = 0$ and $m = 1$} {\label{eco22}
      $f \leftarrow f^2 \overline{L_{-T, -T}}$; $T \leftarrow 2T$; $m \leftarrow 0$ \; 
    }

    \nl \If{$b_i = 1$ and $m = 1$} { \label{eco32}
      $f \leftarrow f^2 L_{2T,P}\overline{L_{-T, -T}}$;  $T \leftarrow 2T + P$;   $m \leftarrow 1$\;
    }

    \nl \If{$b_i = 1$ \text{and} $m = 0$} { \label{parabola1}
      $f \leftarrow f^2 P_{T, P}(x, y)$; $T \leftarrow 2T + P$; $m \leftarrow 1$\;
    }

   }
%  }	
  
  \KwRet{ $f$} 
}

\caption{Improved Refinement of Miller's Algorithm for Even Twisted Curves during Tate pairing computation}
\label{our:algoMiller1}

\end{algorithm}

The factor $P_{T, P}(x, y)$ in the case of line~\ref{parabola1} of the above algorithm is the parabola equation described in Eq.~\ref{eq5}.

Table~\ref{tab:Cmp1} gives a comparison between the modified algorithm and the original Miller's algorithm and Barreto et al.'s algorithm.

\begin{table*}%[htbp]
	\centering
		\begin{tabular}{|l|c|c|}
		\hline
			& Doubling & Doubling and Addition \\
		\hline
		\multirow{2}{*}{Miller's algorithm~\cite{Mil04}} & 2{\bf S} + 2{\bf M} & 2{\bf S} + 4{\bf M}\\
							&	= 3.6{\bf M}				&	= 5.6{\bf M}			 \\
		\hline
		\multirow{2}{*}{Barreto et al.'s algorithm~\cite{BLS03}} & 1{\bf S} + 1{\bf M} & 1{\bf S} + 2{\bf M}\\
																		&	= 1.8{\bf M}				&	= 2.8{\bf M}			 \\
		\hline
		\multirow{2}{*}{{\bf Algorithm~\ref{our:algoMiller1}}} & 1{\bf S} + 1{\bf M} & 1{\bf S} + 2{\bf M} =2.8{\bf M} (line {\bf 3})\\
		&=1.8{\bf M} & 1{\bf S} + 1{\bf M} = 1.8{\bf M} (line {\bf 4})\\
		\hline
		\end{tabular}
	\caption{Number of operations in $\F_{q^k}$ during Tate pairing computation in the case of curves with even embedding degree}
	\label{tab:Cmp1}
\end{table*}

From Table~\ref{tab:Cmp1}, we can see that our algorithm needs no more effort to update the function $f$ than Barreto et al.'s algorithm. When the complexity of operations in $\F_{q^k}$ dominate the complexity of the operations that occur in the lower degree subfields, the total cost of {\bf Algorithm~\ref{our:algoMiller1}} is only about 60\% of that of the original Miller's algorithm.

\section{Experiments}
\label{sec:expe}

We implemented our algorithms and ran some experiments on different elliptic curves at the $128$-bits security level. For this security level, one can choose elliptic curves with the embedding degree $6 \le k \le 10$ when $\rho \approx 2$, and $12 \le k \le 20$ when $\rho \approx 1$ (see \cite[Table 1]{FST10}). In our implementations, we implemented curves with the embedding degrees $k = 9$ \cite{LZZW08}, $k = 12$ \cite{BN06}, and $k = 18$ \cite{KSS08}. We compared the performance of {\bf Algorithm~\ref{algo:Miller}}, {\bf Algorithm~\ref{our:algoMiller}}, and the algorithm proposed in 2008 by \cite{LZZW08} when $k = 9$, while when $k = 12$ and $k = 18$, we compared the performance of {\bf Algorithm~\ref{algo:Miller}}, {\bf Algorithm~\ref{our:algoMiller1}}, and the algorithm in \cite{BLS03}. 

The implementations which are based the library for doing Number Theory (NTL) \cite{Shoup09}, and the GNU Multi-Precision package (GMP) \cite{gmp}, did not use any optimization trick. Computations on $100$ random inputs are performed only on Miller function (without any final exponentiation) in affine coordinates. Average timings are measured on an Intel(R) Core(TM)2 Duo CPU E8500 @ 3.16GHz, 4 GB of RAM under Ubuntu 10.10 32-bit operating system. The experimental results are summarized in  Figure~\ref{table:Timming}.
%Figure~\ref{table:Timming} gives a comparison among the original Miller's algorithm ({\bf Algorithm}~\ref{algo:Miller}), {\bf Algorithm}~\ref{our:algoMiller}, the algorithm in \cite{BLS03} when the embedding degree $k$ even and the algorithm in \cite{LZZW08} when $k = 3^i$. 

We don't apply any twist in implementations. Thus, in the example with $k = 9$, the full extension field $\F_{p^k}$ was generated as $\F_{p^k}/(x^9 + x + 1)$ while when $k = 12$, the full extension field $\F_{p^k}$ was generated as $\F_{p^k}/(x^{12} + 5)$, and  when $k = 18$, the full extension field $\F_{p^k}$ was generated as $\F_{p^k}/(x^{18} + x + 3)$.

Parameters of used elliptic curves are given as follows:

\begin{itemize}
 \item For $k = 9$, the elliptic curve is defined by $E: y^2 = x^3 + 1$ over a finite field of $348$-bits, and 
  \bigskip

    \begin{small}
     $\begin{array}{rrl}
	r & = & 1758592360244376049423345540022962797459 \\ & & 272736402347141193268746504567484534417; \\
	p & = & 3061451105959572350992904218241517192718 \\ & & 315802710560373001011629795786952195361 \\ & &19724392170588602764112177;\\
	\rho & = & 4/3; \\
	
      \end{array}$

    \end{small}
\bigskip
 \item For $k = 12$, the elliptic curve is defined by $E: y^2 = x^3 + 5$ over a finite field of $254$-bits, and 
  \bigskip

    \begin{small}
     $\begin{array}{rrl}

	r & = & 160305690344031282777566882874986495155 \\ & & 10226217719936227669524443298095169537; \\
	p & = & 160305690344031282777566882874986495156 \\ & & 36838101184337499778392980116222246913;\\
	\rho & = & 1; \\
	
      \end{array}$

    \end{small}
   \bigskip
    \item For $k = 18$, the elliptic curve is defined by $E: y^2 = x^3 + 19$ over a finite field of $335$-bits, and 
  \bigskip

    \begin{small}
     $\begin{array}{rrl}
		
	r & = & 10786994225696144150491191871486839136 \\ & & 9781354128945134119237266728176832001; \\
	p & = & 58709285320900073406925617811693805623 \\ & & 56430404913482030243997510873476960788 \\ & &5279673307215755454161141;\\
	\rho & = & 4/3;
	
      \end{array}$

    \end{small}

\end{itemize}

\begin{table*}%[htbp]
\centering
  \begin{tabular}{|c||c|c|c|c|}
    \hline
    $k$ & Miller algorithm & Our algorithms & BLS algorithm \cite{BLS03} & LZZW algorithm \cite{LZZW08} \\
    \hline
    $9$ &  $0.0568 (s)$ & $0.0404 (s)$ & - & $0.0517 (s)$ \\
    \hline
    $12$ & $0.0509 (s)$  & $0.0278 (s)$ & $0.0285 (s)$ & - \\
    \hline
    $18$ & $0.1164 (s)$  & $0.0596 (s)$ & $0.0613 (s)$ & - \\
    \hline

  \end{tabular}

\caption{Timings}
\label{table:Timming}
\end{table*} 

Table~\ref{table:Timming} shows that our refinement is faster about 25\%-40\% in comparison to the original Miller algorithm. It is also faster about 20\% in comparison with the algorithm in \cite{LZZW08} which eliminates denominators using a cubic twist when the embedding degree $k = 9$. When $k$ even, our algorithm is comparable to the algorithm in \cite{BLS03}. Table~\ref{table:Timming} also shows that at 128-bits security level, the Barreto-Naehrig (BN) curve \cite{BN06} over a prime field of size roughly 256-bits with the embedding degree $k = 12$ is the best choice.

\section{Conclusion and open problems}
In this paper we extended the Blake-Murty-Xu's method to propose further refinements to Miller's algorithm which is at the heart of all pairing-based cryptosystems. Our algorithm can eliminate all of vertical lines in the original Miller's algorithm, and so it is \emph{generically} more efficient than the refinements of Blake-Murty-Xu \cite{BMX06} and that of Liu et al.~\cite{LHC07}. We also proposed a variant that can eliminate denominators as in Barreto et al.'s algorithm for computing Tate pairing on even twisted elliptic curves \cite{BLS03}. 

Our improvement works perfectly well for computing both of Weil and Tate pairings over any pairing-friendly elliptic curve. In~\cite{Ver10}, the author introduced the concept of \emph{optimal pairings} which can be computed with $log_2 r/\varphi(k)$ basic Miller iterations. For example, using Theorem 2 of \cite{Hes08}, it should be possible to find an elliptic curve with a prime embedding degree minimizing the number of iterations. We believe that there will be applications in pairing-based cryptography using elliptic curves with embedding degree not being of form $2^i3^j$. Further work is needed to clarify such questions.

\section*{Funding}
A part of this work was done while the first author held post-doc position at the Universit{\'e} de Caen Basse-Normandie. The second author was financed by the National Science Council, Taiwan, R.O.C., under contract number: NSC 98-2221-E-468-012 and NSC 99-2221-E-468-012. 

%\bibliographystyle{plain}
%\bibliographystyle{compj}
%\bibliography{refs}
%\end{document}

\appendix
\section{Algorithm 3 in \cite{BMX06}}
\label{appd1}
In radix-4 representation, Blake, Murty and Xu \cite{BMX06} presented a refinement on Miller algorithm which works as in {\bf Algorithm~\ref{algo:BMX4}}.

\begin{algorithm}
   
   %\SetLine
   \KwIn{  $r = \sum_{i = 0}^t q_i 4^i$, $q_i \in \{0, 1, 2, 3\}$, $P, Q \in E[r]$.}
   \KwOut{ $f$}
  \BlankLine
   { 
   $T \leftarrow P$, $f \leftarrow 1$\;
    \If{$q_r = 2$} {
        $f \leftarrow f^2 \cdot \frac{L_{P,P}(Q)}{V_{2P}(Q)}$; $T \leftarrow 2P$ \; 
    }
    \If{$q_r = 3$} {
        $f \leftarrow f^3 \cdot \frac{L_{P,P}^2(Q) \cdot V_P(Q)}{L_{2P, P}(-Q)}$; $T \leftarrow 3P$ \; 
    }
 \For{$i = t - 1$  \emph{\KwTo}  $0$} 
 {  
    \If{$q_i = 0$} {
      $f \leftarrow f^4 \frac{L_{T,T}^2(Q)}{L_{2T,2T}(-Q)}$; $T \leftarrow 4T$\; 
    }

    \If{$q_i = 1$} {
      $f \leftarrow f^4 \frac{L_{T,T}^2(Q) \cdot L_{4T,P}(Q)}{V_{4T + P}(Q) \cdot L_{2T,2T}(-Q)}$; $T \leftarrow 4T + P$\; 
    }

    \If{$q_i = 2$} {
      $f \leftarrow f^4 \frac{L_{T,T}^2(Q) \cdot L_{2T,P}^2(Q)}{V_{2T}^2(Q) \cdot L_{2T + P,2T + P}(-Q)}$; $T \leftarrow 4T + 2P$\; 
    }

    \If{$q_i = 3$} {
      $f \leftarrow f^4 \frac{L_{T,T}^2(Q) \cdot L_{2T,P}^2(Q) \cdot L_{4T + 2P, P}(Q)}{V_{2T}^2(Q) \cdot L_{2T + P,2T + P}(-Q) \cdot V_{4T + 3P}(Q)}$; $T \leftarrow 4T + 3P$\; 
    }

   }
%  }	
  
  \KwRet{ $f$} 
}

\caption{Blake-Murty-Xu's Refinement on Miller's Algorithm in base 4}
\label{algo:BMX4}

 \end{algorithm}

\end{document}